\documentclass{lmcs}
\pdfoutput=1
\usepackage[utf8]{inputenc}

\usepackage{lastpage}
\lmcsdoi{22}{1}{28}
\lmcsheading{}{\pageref{LastPage}}{}{}%
{Nov.~22,~2024}{Mar.~31,~2026}{}

\usepackage{amsmath,amssymb,amsthm}
\usepackage{xcolor,graphicx}
\usepackage{mathrsfs,mathtools}
\usepackage{csquotes}
\usepackage{enumitem}

\usepackage[english]{babel}
\usepackage{enumitem}
\usepackage{xfrac}
\usepackage{tikz}
\usetikzlibrary{calc, arrows}

\newtheorem{claim}{Claim}

\newcommand{\Jd}{\mathrel{\diamond}}
\newcommand{\sA}{\mathrel{\subseteq_\forall}}

\DeclareMathOperator{\diam}{diam}
\DeclareMathOperator{\mesh}{mesh}
\DeclareMathOperator{\fdiam}{fdiam}
\DeclareMathOperator{\fmesh}{fmesh}

\newcommand{\Xd}{$(X,d\mspace{2mu})$}
\newcommand{\Xdalfa}{$(X,d,\alpha)$}
\newcommand{\R}{\mathbb R}

\newcommand{\N}{\mathbb N}
\newcommand{\eps}{\varepsilon}
\newcommand{\abs}[1]{\left|#1\right|}
\newcommand{\cl}{\operatorname{Cl}}
\newcommand{\kraj}{\overline}
\newcommand{\zadnji}[1]{{(#1)_{\kraj{#1}}}}
\newcommand{\luk}[3]{\overline{#1#2}^{#3}\!}
\newcommand{\ugll}{{[\mspace{2mu}l\mspace{2mu}]}}
\newcommand{\uglj}{{[\mspace{2mu}j\mspace{1mu}]}}
\newcommand{\ugllc}{{[\mspace{2mu}l'\mspace{1mu}]}}
\newcommand{\uglln}{{[\mspace{2mu}l_n\mspace{1mu}]}}
\newcommand{\defj}{:=}
\newcommand{\qand}{\text{\quad and\quad}}
\newcommand{\citeup}[2]{\textup{\cite[#1]{#2}}}
\newcommand{\defn}{\textbf}
\newcommand{\Sjedan}{\mathbb S^1}

\begin{document}

\title[Computable approximations of semicomputable graphs]{Computable approximations\texorpdfstring{\\}{} of semicomputable graphs}

\author{Vedran Čačić\lmcsorcid{0000-0003-2330-9740}}
\author{Matea Čelar\lmcsorcid{0000-0002-3850-1869}}
\author{Marko Horvat\lmcsorcid{0000-0003-0515-8723}}
\author{Zvonko Iljazović\lmcsorcid{0000-0003-0755-8050}}
\email{\quad veky@math.hr,\quad matea.celar@math.hr,\quad mhorvat@math.hr,\quad zilj@math.hr}
\address{University of Zagreb, Faculty of Science, Department of Mathematics}

\begin{abstract}
In this work, we study the computability of topological graphs, which are obtained by gluing arcs and rays together at their endpoints.
We prove that every semicomputable graph in a computable metric space can be approximated, with arbitrary precision, by its computable
subgraph with computable endpoints.
\end{abstract}

\maketitle

\section{Introduction}

Our work progresses the research programme, started by Miller \cite{miller}, that is trying to determine which conditions render semicomputable subsets of computable metric spaces computable.

A compact set $S$ in a computable metric space $(X,d,\alpha )$ is said to be semicomputable if we can effectively enumerate all finite unions of rational open balls in $(X,d,\alpha )$ which cover $S$. In certain computable metric spaces (such as in Euclidean space) this property is equivalent to the fact that the complement $X\setminus S$ of $S$ can be effectively exhausted by rational open balls. On the other hand, a closed subset $S$ of a computable metric space is said to be \emph{computably enumerable} (\emph{c.e.}) if we can effectively enumerate all rational open balls intersecting $S$. A compact set $S$ is said to be computable if $S$ can be effectively approximated by a finite set of rational points with any given precision. It can be shown that a subset is computable if and only if it is semicomputable and c.e.

Although each nonempty computable set contains computable points (moreover, they are dense in it), there is a nonempty semicomputable subset of $\mathbb{R}$ which does not contain any computable point. So a semicomputable set can be thought of as being ``far from computable''. While we proceed from the assumption of semicomputability, one might also ask whether every compact set is homeomorphic to a computable set. The answer is negative, moreover there is a c.e. closed subset of $[0,1]$ which is not homeomorphic to any computable compact subset of any computable metric space~\cite{badaev2025non, bosserhoff2025}.

However, there are certain topological conditions under which semicomputable sets are in fact computable. It was shown in \cite{miller} that each semicomputable set homeomorphic to a sphere is computable. Moreover, every semicomputable manifold is computable \cite{lmcs:mnf, manifolds-top}. Many other results on this topic can be found in \cite{ah, ah2022, csoschc, kihara, compintpoint, lmcs:1mnf}.

On the other hand, even a semicomputable set of a simple topological type such as a semicomputable arc (i.e.\ a space homeomorphic to $[0,1]$) need not be computable. Natural questions arise if $S$ is semicomputable: does $S$ contain a computable point? Can $S$ be approximated by a computable subset of a certain type?

Kihara constructed in \cite{kihara}, as the
answer to a question in \cite{ziegler}, 
a nonempty semicomputable compact set in the plane which is simply
connected (in fact, contractible) and contains no computable points.

In the case of an arc, the following result is known: if $S$ is a semicomputable arc with endpoints $a$ and $b$, then there exist computable points $a'$ and $b'$ in $S$ arbitrarily close to $a$ and $b$, respectively, and a subarc $T$ of $S$ whose endpoints are $a'$ and $b'$ and such that $T$ is a computable set \cite{jucs, compintpoint}. Moreover, a similar result holds for decomposable chainable continua which are spaces more general than arcs \cite{jucs, compintpoint}.

In this paper, we consider \emph{(topological) graphs}: spaces obtained by gluing arcs along their vertices. A semicomputable set $S$ is computable if $S$ is a topological graph such that each of its endpoints is computable \cite{mlq-graphs}.

In general, however, a semicomputable graph need not be computable. So it makes sense to ask the following question: if $S$ is a semicomputable graph, can $S$ be approximated by a computable subgraph? We prove that each semicomputable graph $S$ can be approximated by a computable subgraph which is obtained by cutting off small subarcs around all the uncomputable endpoints of $S$.

The notion of a semicomputable set can be extended to sets which are not necessarily compact. For example, each semicomputable $1$-manifold (not necessarily compact) with finitely many connected components is computable \cite{lmcs:1mnf}. In \cite{mlq-graphs} the notion of a \emph{noncompact graph} was investigated and it was proved that semicomputability of such a space implies its computability under the assumption that each of its endpoints is computable.
In this paper we prove a result analogous to the result for graphs: each semicomputable noncompact graph $S$ can be approximated by a computable noncompact subgraph which is obtained by cutting off small subarcs around all the uncomputable endpoints of $S$.

If $S$ is a semicomputable set in a computable metric space and $x\in S$ such that $x$ has a neighborhood in $S$ homeomorphic to $\mathbb{R}^{n} $ for some $n\in \N\setminus \{0\}$, then $x$ has a neighborhood in $S$ which is a computable set. This result was proved in \cite{apal}. In the case when $n=1$ we prove a stronger result: $x$ has a computable neighborhood in $S$ which is an arc with computable endpoints. This fact is crucial in the proofs  of the main results of the paper.

It should be mentioned that replacing our setting of computable metric spaces with that of computable Hausdorff spaces would not lead to  more general results. Recently,
it was shown~\cite[Theorem 3.4]{ah} that every semicomputable set $S$ in a computable Hausdorff space can be effectively embedded into the Hilbert cube, i.e.\ there is a computable homeomorphism $f$ with a computable inverse such that $f(S)$ is semicomputable and for each computable subset $T$ of the Hilbert cube, $f^{-1}(T)$ is computable.

\section{Background}

\subsection{Computable metric spaces}\label{ssec:compmet}

In this subsection we provide some basic facts about computable metric spaces. See~\cite{per,we,turing,we2,br-we,br-pr,jucs,manifolds-top,dowmel,IK21}.

Let $k\in\N\setminus\{0\}$. A function $f:\N^k
\to\mathbb Q$ is said to be \defn{computable} if there
exist computable (i.e.\ recursive) functions $a,b
,c:\N^k\to\N$ such that
\begin{equation*}
f(x)=(-1)^{c(x)}\cdot\dfrac{a(x)}{b(x)+1}
\end{equation*}
for each $x\in\N^k$.
A function $f:\N^k\to\R$ is said to be \defn{computable} if there exists a computable function $F:\N^{k+1}\to\mathbb Q$ such that
\[\abs{f(x)-F(x,i)}<2^{-i}\]
for all $x\in\N^k$, $i\in\N$.

Let \Xd\ be a metric space and let $\alpha $ be a sequence in $X$ such that $\alpha(\N)$ is a dense set in \Xd. We say that \Xdalfa\ is a \defn{computable metric space} if the function $(i,j)\mapsto d(\alpha_i,\alpha_j):\N^2\to\R$ is computable.

For example, if $n\in\N\setminus\{0\}$, $d$ is the Euclidean metric on $\R^n$ and $\alpha:\N\to\mathbb Q^n$ is an effective enumeration of $\mathbb Q^n$, then $(\R^n,d,\alpha)$ is a computable metric space. This space is called \defn{computable Euclidean space}.

Let \Xdalfa\ be a computable metric space. A point $x \in X$ is said to be \defn{computable} in \Xdalfa\ if there is a computable function $f : \N \to \N$ such that 
\[d\bigl(x, \alpha_{f(k)}\bigr) < 2^{-k}\]
for all $k \in \N$.

From now on, let $(j, i) \mapsto (j)_i:\N^2 \to \N$ and $j \mapsto \kraj j:\N \to \N$ be fixed computable functions such that 
\[  \bigl\{ \bigl( (j)_0, (j)_1, \dots, \zadnji j\bigr) \bigm| j \in \N \bigr\}\]
is the set of all nonempty finite sequences in $\N$.

For $j \in \N$ let 
\[ \uglj \defj \bigl\{ (j)_0,\dots,\zadnji j\bigr\} \text.\]
Then each nonempty finite subset of $\N$ is equal to $\uglj$ for some $j \in \N$.

Let \Xd\ be a metric space, $A, B \subseteq X$ and $\eps > 0$. We say that $A$ and $B$ are \defn{$\eps$-close}, and write $A\approx_\eps\!\! B$, if
\[(\forall a \in A)(\exists b \in B)\bigl(d(a, b) < \eps\bigr) \qand (\forall b \in b)(\exists a \in A)\bigl(d(a, b) < \eps\bigr) \text.\]

If $A$ and $B$ are nonempty compact sets in \Xd, the number $\inf\{\eps >0\mid A\approx_\eps\! B\}$ is called the \defn{Hausdorff distance} from $A$ to $B$ and it is denoted by $d_{H}(A,B)$.

It is not hard to check that, for $\eps >0$, we have $d_{H}(A,B)<\eps $ if and only if $A\approx_\eps\! B$.

\begin{defi}
    Let \Xdalfa\ be a computable metric space. We say that a compact set $S \subseteq X$ is \defn{computable} in \Xdalfa\ if $S=\emptyset$ or there is a computable function $f : \N \to \N$ such that
    \[S \approx_{2^{-k}}\! \{\alpha_i \mid i \in [f(k)] \},\quad \text{for all } k\in\N\text.\]
\end{defi}

Let \Xdalfa\ be a computable metric space. Let $\tau_1,\tau_2:\N\to\N$ be fixed computable functions such that $\bigl\{\bigl(\tau_1(i),\tau_2(i)\bigr)\bigm| i\in\N\bigr\}=\N^2$ and let $q:\N\to\mathbb Q$ be an effective enumeration of all positive rational numbers. For $i\in\N$ let
\[I_i\defj B\!\left(\alpha_{\tau_1(i)},q_{\tau_2(i)}\right)\! \qand \hat I_i \defj \hat B\!\left(\alpha_{\tau_1(i)},q_{\tau_2(i)}\right)\!\text,\]
where $B\!\left(\alpha_{\tau_1 (i)},q_{\tau_2(i)}\right)\!$ is the open ball in \Xd\ with radius $\rho_i\defj q_{\tau_2(i)}$ centered at $\lambda_i\defj\alpha_{\tau_1(i)}$, and $\hat B\!\left(\lambda_i,\rho_i\right)\!$ is the corresponding closed ball.

For $j,l \in \N$ let
\[ J_j \defj \bigcup\nolimits_{i \in \uglj} I_i,\quad \hat J_j \defj \bigcup\nolimits_{i \in \uglj} \hat I_i \qand  J_\ugll \defj \bigcup\nolimits_{j \in \ugll} J_j\text.\]
Then $(J_j)$ and $(\hat J_j)$ are, respectively, effective enumerations of all finite unions of rational open and closed balls.

\begin{defi}\label{def:ce-semicomp-cpt}
Let \Xdalfa\ be a computable metric space.
\begin{enumerate}[label=(\roman*)]
    \item\label{it:ce} A closed set $S$ in \Xd\ is said to be \defn{computably enumerable} (c.e.)\ in \Xdalfa\\ if the set $\{i\in\N\mid I_i \cap S\ne \emptyset\}$ is c.e.
    \item\label{it:semicomp-cpt} A compact set $S$ in \Xd\ is said to be \defn{semicomputable} in \Xdalfa\\ if the set $\{j\in \N\mid S\subseteq J_{j}\}$ is c.e.
\end{enumerate}
\end{defi}
It is not hard to see that these definitions do not
depend on the choice of particular functions $q$, $\tau_1$, $\tau_2$ and $j\mapsto\uglj$. Semicomputable sets are also called \textit{computably compact} or \textit{effectively compact} in the literature \cite{dowmel}.

We write $K\subseteq_\eps J_j$ if:
\begin{itemize}
\item $K\subseteq J_j$,
\item for all $i\in\uglj$, $I_i\cap K\ne\emptyset$, and
\item for all $i\in\uglj$, $\rho_i<\eps$.
\end{itemize}
Note that $K \subseteq_\eps J_j$ implies $K \approx_\eps\! \{ \lambda_i \mid i \in [j] \}$.

Among multiple  characterizations of a computable compact set, we highlight the following one~\cite[Proposition 2.6]{lmcs:mnf}:
\begin{equation}\label{eq:compset-ce-semi}
S\text{ computable} \quad\Longleftrightarrow\quad S\text{ c.e.\ and semicomputable.}
\end{equation}

The notion of semicomputability can be generalized to noncompact sets~\cite{lmcs:1mnf,manifolds-top}.
\begin{defi}
Let \Xdalfa\ be a computable metric space.\\ A set $S \subseteq X$ is \defn{semicomputable} in \Xdalfa\ if the following hold:
\begin{enumerate}
\item[(i)] $S\cap B$ is a compact set for each closed ball $B$ in \Xd;
\item[(ii)] the set $\{(i,j)\in \N^2\mid \hat I_i \cap S\subseteq J_j\}$ is c.e.
\end{enumerate}  
\end{defi}

This notion extends the previously defined notion of semicomputability in the following sense: if $S$ is a compact set in \Xd, then $S$ is semicomputable in the sense of the latter definition if and only if $S$ is semicomputable in the sense of Definition~\ref{def:ce-semicomp-cpt}\ref{it:semicomp-cpt}. 

Additionally, note that semicomputability, in our general sense, differs from the seemingly related notion of effective local compactness~\cite{elc} used by Pauly. A distinguishing example are the natural numbers with the discrete metric --- this space is not semicomputable, however it is effectively locally compact. 

The notion of a computable set in a computable metric space can be now extended to noncompact sets in the obvious way: $S$ is \defn{computable} in \Xdalfa\ if $S$ is semicomputable and c.e.\ in \Xdalfa.

Each semicomputable set $S$ is co-c.e., which means that $X\setminus S=\bigcup_{i\in A}I_i$ for some c.e.\ set $A\subseteq\N$. On the other hand, a co-c.e.\ set need not be semicomputable, but the equivalence
\[S\text{ semicomputable}\Longleftrightarrow S\text{ co-c.e.}\]
holds in some computable metric space, for example in computable Euclidean space $(\R^n,d,\alpha)$.

\begin{defi}
    Let $i, j \in \N$. 
    \begin{itemize}
        \item We say that $I_i$ and $I_j$ are \defn{formally disjoint} if
\[d(\lambda_i, \lambda_j) > \rho_i + \rho_j \text.\]
We denote this by $I_i \Jd I_j$.
    \item We say that $J_i$ and $J_j$ are \defn{formally disjoint} if $I_k \Jd I_l$ for all $k \in [i]$ and $l \in \uglj$. We denote this by $J_i \Jd J_j$.
    \end{itemize}  
\end{defi}
Note that $I_i \Jd I_j$ implies $\hat I_i \cap \hat I_j = \emptyset$. Similarly, $J_i \Jd J_j$ implies $\hat J_i \cap \hat J_j = \emptyset$.

\begin{defi}
     Let $i, j \in \N$. 
     \begin{itemize}
         \item We say that $I_i$ is \defn{formally contained} in $I_j$ and we write $I_i \sA I_j$ if 
        \[d(\lambda_i, \lambda_j) + \rho_i < \rho_j \text.\]
        \item We say that $J_{i}$ is \defn{formally contained} in $J_{j}$ and we write $J_i \sA J_j$ if 
        \[(\forall k \in [i])(\exists l \in [j])(I_k \sA I_l)\text.\]
        \item We say that $J_{[i]}$ is \defn{formally contained} in $J_{[j]}$ and we write $J_{[i]} \sA J_{[j]}$ if 
        \[(\forall k \in [i])(\exists l \in [j])(J_k \sA J_l)\text.\]
     \end{itemize}
\end{defi}

Note that $I_i \sA I_j$ implies $\hat I_i \subseteq I_j$ and $J_i \sA J_j$ implies $\hat J_i \subseteq J_j$.

Note also that formal containment and formal disjointness are relations between numbers $i$ and $j$ (and not between the sets $I_i$ and $I_j$ or $J_i$ and $J_j$), and that they are c.e.

Let \Xdalfa\ be a computable metric space. For every $j\in\N$ we define
\begin{equation}
\fdiam(j)\defj\diam\,\{\lambda_u\mid u\in\uglj\}+2\max\,\{\rho_u\mid u\in\uglj\}\text,
\end{equation} and call it the \defn{formal diameter} of $J_j$. Again, this is formally a function of $j$, not of $J_j$.

It is easy to prove that the function $\fdiam:\N\to\R$ is computable~\cite[Proposition 13]{jucs}.

\subsection{Chains}

Let \Xd\ be a metric space. A finite sequence $\mathcal C = (C_0, \dots, C_n)$ of nonempty subsets of $X$ is said to be a
\begin{itemize}
    \item \defn{quasi-chain} in $X$ if, for all $i,j$, $C_i\cap C_j\ne\emptyset$ implies $\abs{i-j}\le 1$;
    \item \defn{chain} in $X$ if,
    for all $i,j$, $C_i\cap C_j\ne\emptyset$ is equivalent to $\abs{i-j}\le 1$;
\end{itemize}
Each set $C_i$ is said to be the ($i$th) \defn{link} of the (quasi-\!)chain $\mathcal C$. If $i < j$, we say that the link $C_i$ \defn{precedes} $C_j$, and if $i + 1 < j$, we say that $C_i$ \defn{strictly precedes} $C_j$.

A (quasi-\!)chain is said to be \defn{open} in \Xd\ if each of its links is open in \Xd. Similarly, a (quasi-\!)chain is said to be \defn{compact} if each of its links is compact.

Suppose $S$ is a subset of $X$ and $\mathcal C = (C_0, \dots, C_n)$ is a (quasi-\!)chain in $X$. We say that $\mathcal C$ \defn{covers} $S$ if $S \subseteq C_0 \cup \dots \cup C_n$.

Suppose $\mathcal A = (A_0, \dots, A_n)$ is a finite sequence of nonempty bounded subsets of $X$. The \defn{mesh} of $\mathcal A$ is the number defined by
\[ \mesh (\mathcal A) \defj \max_{i=0}^n\,\diam A_i \text.\]

If $\mathcal C$ is a chain and $\mesh(\mathcal C) < \eps$, we say that $\mathcal C$ is an $\eps$-\defn{chain}. We say that a continuum (i.e.\ a compact and connected metric space) $X$ is \defn{chainable} if there is an open $\eps$-chain in $X$ covering $X$ for every $\eps > 0$. For $a, b \in X$, we say that a continuum $X$ is \defn{chainable from} $a$ to $b$ if, for every $\eps > 0$, there is an open $\eps$-chain $(U_0, \dots, U_n)$ in $X$ covering $X$ such that $a \in U_0$ and $b \in U_n$. 

Let $\mathcal C = (C_0,\dots,C_n)$ and $\mathcal D = (D_0,\dots,D_m)$ be two (quasi-\!)chains in $X$. We say that $\mathcal D$ \defn{refines} $\mathcal C$ if for each $i \in  \{0, \dots, m\}$ there exists $j \in \{0,\dots, n\}$ such that $D_i \subseteq C_j$. If $\mathcal D$ \defn{refines} $\mathcal C$ and $D_0 \subseteq C_0$ and $D_m \subseteq C_n$, we say that $\mathcal D$ \defn{strongly refines} $\mathcal C$.

Let \Xdalfa\ be a computable metric space and let $l \in\N$. We say that $\bigl( J_{j_0}, \dots, J_{j_m})$ is a \defn{formal chain} in \Xdalfa\ if $J_{j_u} \diamond J_{j_v}$
for all $u, v \in \{0, \ldots, m\}$ such that $\abs{u-v} > 1$. Note that every formal chain is a quasi-chain.

\subsection{Graphs}

Let $n\in \N\setminus \{0\}$ and let $a,v\in \R^n $, $v\ne 0$. The set $\{a+tv\mid t\in [0,\infty)\}$ is called a \defn{ray} in $\R^n $ and we say that $a$ is the (only) \defn{endpoint} of this ray. Let $K$ be a nonempty finite family of line segments and rays in $\R^n$ such that
\begin{equation}\label{eq:simplcompl-2}
\mbox{\parbox{10em}{
\centering
$I, J\in K$ are such that  \\ $I\ne J$
and $I\cap J\ne\emptyset$}}
\quad\implies\quad \mbox{\parbox{14em}{ 
$I\cap  J =\{a\}$, where $a$ is an \\
endpoint of both $I$ and $J$.
}}    
\end{equation}
Then any topological space $G$ homeomorphic to $\bigcup K$ is called a \defn{graph}.

Let $G$ be a graph and $x\in G$. We say that $x$ is an \defn{endpoint} of $G$ if $x$ has an open neighbourhood in $G$ which is homeomorphic to $[0,\infty\rangle $ by a homeomorphism which maps $x$ to $0$.

\begin{figure}[ht]
    \centering
    \begin{tikzpicture}
\draw[<-] (-2,6) -- (2,6) -- (3.5,6.5) -- (3,5) -- (4, 4.5);
\fill (3.5,7.5) circle (0.06);
\draw (2,6) -- (3,5);
\draw (5,6.5) -- (3.5, 7.5);
\fill (4,4.5) circle (0.06);
\fill (4.5,5) circle (0.06);

\draw (4.5, 5) -- (5, 6.5);
\draw[->] (5,6.5) -- (8.5,7);
\end{tikzpicture}
    \caption{A graph. Filled circles denote endpoints and arrows denote rays.}
    \label{fig:generalized-graph}
\end{figure}

Suppose $G$ is a graph and $f:\bigcup K\to G$ a homeomorphism, where $K$ is the family from the definition of a graph. Then:
\begin{itemize}
    \item $x\in G$ is an endpoint of $G$ if and only if there exists $y\in \bigcup K$ such that $f(y)=x$ and such that $y$ is an endpoint of a unique element of $K$;
    \item $G$ is compact if and only if $K$ does not contain rays.
\end{itemize}

Note that our naming scheme differs slightly from \cite{mlq-graphs}; what was previously called a \emph{generalized graph} we call a \emph{graph}, and what was previously called a \emph{graph} is now a \emph{compact graph}.

\subsection{Auxiliary results}
We will use the following results from the literature:
\begin{lemC}[\citeup{Lemma 3.8}{compintpoint}]\label{lm:links}
	Let $X$ be a set and let $(C_0,\dots,C_m)$ and
	$(D_0,\dots,D_{m'})$ be two quasi-chains in $X$ such that
	$(D_0,\dots,D_{m'})$ refines $(C_0,\dots,C_m)$.\\
	Suppose that $i,j,k\in\left\{0,\dots,m\right\}\!$ and $p,q\in
	\left\{0,\dots,m'\right\}\!$ are such that
	\[i<k<j\text,\quad p<q\text, \quad D_p\subseteq C_i \qand D_q\subseteq C_j\text.\]
Then there exists $r\in\left\{0,\dots,m'\right\}$ such that $p<r<q$ and $D_r\subseteq C_k$.
\end{lemC}

\begin{lemC}[\citeup{Lemma 41}{jucs}]\label{lm:presjek-lanaca} 
Let \Xd\ be a metric space in which every closed ball is compact.
Let $\mathcal C^k =(C^k_0,\dots,C^{k}_{m_k})$, $k\in \N$ be a sequence of chains such that for all $k\in\N$, $\bigl( \cl (C^{k+1}_{0}),\dots,\cl(C^{k+1}_{m_{k+1}}) \bigr)$ strongly refines $\left( C^k_0,\dots,C^k_{m_k}\right)$ and $\mesh(\mathcal C^k)<2^{-k}$.
Let 
\[S=\bigcap_{k\in \N}\left(\cl(C^{k+1}_0) \cup \dots \cup \cl(C^{k+1}_{m_{k+1}})\right)\text.\]
Then $S$ is a continuum, chainable from $a$ to $b$, where $a\in \bigcap_{k\in\N} C^{k}_{0}$ and
$b\in \bigcap_{k\in\N} C^k_{m_k}$.
\end{lemC}

\begin{thmC}[\citeup{Theorem 5.2}{apal}]\label{tm:IV-5-2}
    Suppose \Xdalfa\ is a computable metric space, $S$ a semicomputable set in this space and $x \in S$ a point which has a neighbourhood in $S$ homeomorphic to $\R^n$ for some $n \in \N \setminus\{0\}$. Then $x$ has a computable compact neighbourhood in $S$.
\end{thmC}

Theorem~\ref{tm:IV-5-2} essentially follows from the following two results: 
\begin{enumerate}
\item[(i)] if a set $S$ is compact and computable at $x$, then $x$ has a computable compact neighbourhood in $S$ \citeup{Theorem 3.7}{apal};
\item[(ii)] if $S$ is a semicomputable compact set and $x\in S$ has a neighbourhood in $S$ which is homeomorphic to $\R^n$ for some $n \in \N \setminus\{0\}$, then $S$ is computable at $x$ \citeup{Theorem~5.6}{lmcs:mnf}.
\end{enumerate}    
        That $S$ is computable at $x$ means that there exist a neighbourhood $N$ of $x$ in $S$ and a computable function $f:\N\to\N$ such that for each $k\in \N$ the following holds:
\begin{align*}
    \mbox{for each }x\in N\mbox{ there exists }i\in [f(k)]\mbox{ such that }d(x,\alpha _{i} )<2^{-k};\\
    \mbox{for each }i\in [f(k)]\mbox{ there exists }x\in S\mbox{ such that }d(\alpha _{i},x )<2^{-k}.
\end{align*}

The proof of result (ii) in the case $n=1$ relies on the following topological fact related to connectedness: if $a,b\in \mathbb{R}$, $a<b$, $f:[a,b]\to X$ is a continuous injection and $C_{0} ,\dots ,C_{m} $ is an open chain in $(X,d)$ such that $f([a,b])\subseteq C_{0} \cup \dots \cup C_{m} $, $f(a)\in C_{0} $ and $f(b)\in C_{m} $, then $C_{i} \cap f([a,b])\neq\emptyset$ for each $i\in \{0,\dots ,m\}$.

Let us mention here that chains have an important role in the investigation of effective properties of arcs (and chainable continua in general), see \cite{KOH_MELNIKOV_NG_2025, harrison2024arithmetic, jucs, lmcs:1mnf, compintpoint, tocs-20}.

\begin{thmC}[\citeup{Theorem 6.5}{mlq-graphs}]\label{tm:grafovi} Let \Xdalfa\ be a computable metric space and let $S$ be a semicomputable set in this space. Suppose $S$, as a subspace of \Xd, is a graph such that the set $E$ of all endpoints of $S$ is semicomputable in \Xdalfa. Then $S$ is computable in \Xdalfa.
\end{thmC}

\begin{propC}[\citeup{proposition 12}{jucs}]\label{prop:fdiam}
Let \Xdalfa\ be a computable metric space.
\begin{enumerate}[label=(\roman*)]
\item\label{it:fdiam1} For all $j\in\N$, $\diam \widehat{J}_{j}\le \fdiam(j).$
\item\label{it:fdiam2} Let $K$ and $U$  be subsets of \Xd\ such that $K$ is nonempty and compact, $U$ is open and $K\subseteq U.$ Let $\eps >0.$ Then there exists $j\in\N$ such that $K\subseteq J_j$, $\widehat J_j \subseteq U$  and  $\fdiam(j)< \diam K+\eps$.
\end{enumerate}
\end{propC}

\begin{lem}\label{lm:S-minus-J-semicomp}
Let \Xdalfa\ be a computable metric space and let $S$ be a semicomputable set in this space (not necessarily compact). Let $m\in\N$. Then $S\setminus J_m$ is a semicomputable set. 
\end{lem} 
\begin{proof}
Let $\Omega =\{(i,j)\in \N^2\mid \hat I_i \cap S\subseteq J_j\}$.

If $B$ is a closed ball in \Xd, then $(S\setminus J_{m})\cap B$ is a closed set in \Xd\ contained in $S\cap B$. Since $S\cap B$ is compact, $(S\setminus J_{m})\cap B$ is compact too.

Let $i,j\in \N$. We have
\[\hat I_i \cap (S\setminus J_{m})\subseteq J_j\Longleftrightarrow (\hat I_i \cap S)\setminus J_{m} \subseteq J_j \Longleftrightarrow \hat I_i \cap S\subseteq J_j\cup  J_{m}\text.\]
Choose a computable function $f:\N^2\to\N$ such that $J_{a}\cup J_{b}=J_{f(a,b)}$ for all $a,b\in\N$. Then
\[\hat I_i \cap (S\setminus J_{m})\subseteq J_j\Longleftrightarrow \hat I_i \cap S\subseteq J_{f(j,m)}\Longleftrightarrow\bigl(i,f(j,m)\bigr)\in \Omega\text.\]
So the set of all $(i,j)\in \N^2$ such that $\hat I_i \cap (S\setminus J_{m})\subseteq J_j$ is the set of all $(i,j)\in \N^2$ such that $\bigl(i,f(j,m)\bigr)\in \Omega$, which is c.e.\ since $\Omega $ is c.e.\ and $f$ is computable.
\end{proof}

\begin{cor}\label{cor:exists-S'}
    Let \Xdalfa\ be a computable metric space, let $S$ be a semicomputable set in this space and let $K\!$ and $U\!$ be subsets of $S$ such that $K$ is compact, $U$ is open in $S$ and $K \subseteq U$. Then there exists a semicomputable compact set $S' \subseteq S$ such that $K \subseteq S' \subseteq U$. 
\end{cor}
\begin{proof}
    Since $K$ is compact, it is fully contained in a rational closed ball with a sufficiently large diameter. Let $i \in \N$ be such that $K \subseteq L\defj \hat I_i \cap S$.

    Consider the set $L \setminus U$. If it is empty, then $L \subseteq U$, so $S' \defj L$ is a semicomputable compact set such that $K \subseteq S' \subseteq U$.

    Suppose $L \setminus U$ is not empty. Then it is a closed subset of a compact set $L$, therefore it is compact, and $L \setminus U \subseteq L \setminus K$.
    Since $L \setminus K$ is open in $L$, let $V$ be an open set in \Xd\ such that $V \cap L = L \setminus K$. Now, $L \setminus U$ and $V$ are subsets of \Xd\ such that $L \setminus U$ is nonempty and compact, $V$ is open and $L \setminus U \subseteq V$. By Proposition~\ref{prop:fdiam}\ref{it:fdiam2}, there exists $j \in \N$ such that $L \setminus U \subseteq J_j \subseteq V$.
Then $S' \defj L \setminus J_j$ is semicomputable compact and
\begin{gather*}
    L \setminus V \subseteq S' \subseteq L\setminus \bigl( L \setminus U\bigr) \text.\\
\shortintertext{Since}
    K = L \setminus \bigl(L \setminus K\bigr) = L \setminus \bigl(V \cap L\bigr) = L \setminus V\qand
    L\setminus \bigl( L \setminus U\bigr) \subseteq U \text,
\end{gather*}
we conclude that $K \subseteq S' \subseteq U$.
\end{proof}

\begin{propC}[\citeup{Lemma 4.6}{lmcs:1mnf}]\label{prop:separator}
Let \Xdalfa\ be a computable metric space and let
$A,B$ be disjoint nonempty compact sets in \Xd. Then
\begin{enumerate}[label=(\roman*)]
    \item\label{it:sep1} For each $\eps > 0$ there exists $j \in \N$ such that $A \subseteq_\eps J_j$.
    \item\label{it:sep2} There exists $\mu >0$ such that for all $i,j\in\N$, if $A\subseteq_\mu J_i$ and $B\subseteq_\mu J_j$, then $J_i\Jd J_j$.
\end{enumerate}
\end{propC}

\begin{lemC}[\citeup{Lemma 4.8}{lmcs:1mnf}]\label{lm:r-diam}
    Let \Xdalfa\ be a computable metric space. Let $A \subseteq X$, $j \in \N$ and $r > 0$ be such that $A \subseteq_r J_j$. Then
    \[ \fdiam(j) < 4r + \diam A \text.\]
\end{lemC}

\begin{propC}[\citeup{Proposition 4.13}{compintpoint}]\label{prop:augmentator}
Let \Xdalfa\ be a computable metric space, $A\subseteq\N$ a c.e.\ set, $K$ a nonempty compact set such that $K\subseteq\bigcup_{i\in A} I_i$.
 Then there exists
$\mu >0$ such that for all $i\in\N$, if $K\subseteq_\mu J_j$, then
\[(\forall u \in [j])(\exists v \in A)(I_u \sA I_v) \text.\]
\end{propC}

\begin{propC}[\citeup{Proposition 6.5.}{compintpoint}]\label{prop:profinjenje}
    Let \Xdalfa\ be a computable metric space. Let $l, l', u, v \in \N$. The following statements hold:
\begin{enumerate}[label=(\roman*)]
    \item\label{it:prof1} if $J_u \sA J_v$, then $\cl (J_u) \subseteq J_v$;
    \item\label{it:prof2} if $J_\ugll \sA J_\ugllc$, then the finite sequence $\bigl(\cl (J_{(l)_0}), \dots, \cl (J_\zadnji l) \bigr)$ refines the finite sequence $\bigl(J_{(l')_0}, \dots, J_\zadnji{l'}\bigr)$.
\end{enumerate}
\end{propC}

\section{Existence of a computable neighbourhood}

The main result of this section is Theorem~\ref{tm:compneigh} which states that any point in a semicomputable set which has a neighbourhood $N$ homeomorphic to $\R$ also has a computable neighbourhood $N' \subseteq N$ which is an arc with computable endpoints.

Note that Theorem~\ref{tm:IV-5-2} has a similar statement: any point $x$ in a semicomputable set which has a neighbourhood homeomorphic to $\R^n$ also has a computable neighbourhood. However, it does not guarantee that this computable neighbourhood will be an arc with computable endpoints. Since computable points are dense in computable sets, there is a neighbourhood of $x$ which is an arc with computable endpoints, but such neighbourhood is not necessarily computable (or even semicomputable).

In order to construct a neighbourhood which is simultaneously computable and an arc with computable endpoints, our approach is to look at the intersection of a sequence of chains whose links have strictly decreasing diameters.

We begin by listing some important auxiliary results. 
\begin{lem}\label{lm:augmentator-Jj}
Let \Xdalfa\ be a computable metric space, $K$ a nonempty compact set,\\ and $a\in\N$ such that $K\subseteq J_a$.

Then there exists $\mu>0$ such that for all $j\in\N$, if $K\subseteq_\mu J_j$, then $J_j\sA J_a$.
\end{lem}
\begin{proof}
    This follows immediately from Proposition~\ref{prop:augmentator} for $A = [a]$.
\end{proof}

\begin{lem}\label{lm:napuhavanje}
Let \Xdalfa\ be a computable metric space, $(K_0, \dots, K_{n+1})$ a compact quasi-chain in \Xd, and $A\subseteq\N$ a finite set. Then for all $\eps>0$, there exist $p, l, q\in\N$ such that $\kraj l=n-1$ and
\begin{enumerate}[label=(\alph*)]
\item\label{it:napuh-a} $K_0\subseteq_\eps J_p$, $K_{n+1}\subseteq_\eps J_q$ and
$K_i\subseteq_\eps J_{(l)_{i-1}}$ for all $i \in \{1, \dots, n\}$;
\item\label{it:napuh-b} $\bigl(J_p, J_{(l)_0}, \dots, J_{\zadnji l}, J_q\bigr)$ is a formal chain;
\item\label{it:napuh-c} for all $a\in A$, $K_0\subseteq J_a$ implies $J_p\sA J_a$, $K_{n+1}\subseteq J_a$ implies $J_q\sA J_a$ and $K_i\subseteq J_a$ implies $J_{(l)_{i-1}}\sA J_a$, for all $i \in \{1, \dots, n\}$.
\end{enumerate}
\end{lem}

\begin{proof}
    Fix $i, j \in \{0, \dots, n+1\}$ such that $\abs{i-j} > 1$.
    Since $(K_0, \dots, K_{n+1})$ is a compact quasi-chain, $K_i$ and $K_j$ are disjoint nonempty compact sets in \Xd. By Proposition~\ref{prop:separator}\ref{it:sep2}, there exists $\mu_{i, j} > 0$ such that for any $\alpha, \beta \in \N$,
    \begin{equation}\label{eq:napuhavanje-disj}
        K_i \subseteq_{\mu_{i, j}} J_\alpha \qand K_j \subseteq_{\mu_{i, j}} J_\beta \Longrightarrow J_\alpha \Jd J_\beta \text.
    \end{equation}

    Now fix $i \in \{0, \dots, n+1\}$ and $a \in A$ such that $K_i \subseteq J_a$. By Lemma~\ref{lm:augmentator-Jj}, there exists $\mu_{i, a} > 0$ such that for any $\alpha \in \N$,
    \begin{equation}\label{eq:napuhavanje-cont}
        K_i \subseteq_{\mu_{i, a}} J_\alpha \Longrightarrow J_\alpha \sA J_a \text. 
    \end{equation}

    Let 
    \begin{multline*}
     r \defj \min \bigl( \bigl\{\mu_{i, j} \bigm| i, j \in \{0, \dots, n+1\},\ \abs{i-j} > 1\bigr\} \ \cup  \\
     \bigl\{ \mu_{i, a} \bigm| i \in \{0, \dots, n+1\},\ a \in A,\ K_i \subseteq J_a\bigr\} \cup \{ \eps\} \bigr)   
    \end{multline*}

    By Proposition~\ref{prop:separator}\ref{it:sep1}, there exist $p, j_1, \dots j_n, q \in \N$ such that
    \begin{equation}\label{eq:napuhavanje-exp}
        K_0 \subseteq_r J_p,\ K_1 \subseteq_r J_{j_1}, \dots, K_n \subseteq_r J_{j_n} \text{ and } K_{n+1} \subseteq_r J_q \text.
    \end{equation}

    Let $l \in \N$ be such that $(j_1, \dots, j_n) = \bigl((l)_0, \dots, \zadnji l\bigr)$.
    Since $K \subseteq_{\lambda'} J_j$ implies $K \subseteq_\lambda J_j$ for any $\lambda' < \lambda$, we can now easily conclude that~\ref{it:napuh-a} follows from~\eqref{eq:napuhavanje-exp},~\ref{it:napuh-b} follows from~\eqref{eq:napuhavanje-disj}, and~\ref{it:napuh-c} follows from~\eqref{eq:napuhavanje-cont}.
\end{proof}

Any topological space homeomorphic to $[0,1]$ is called an \defn{arc}. If $A$ is an arc and $x\in A$ a point such that $A\setminus \{x\}$ is connected, then we say that $x$ is an \defn{endpoint} of $A$. Note that if $f:[0,1]\to A$ is a homeomorphism, then $f(0)$ and $f(1)$ are the only endpoints of $A$.

\begin{thm}\label{tm:compneigh}
    Let $S$ be a semicomputable set in a computable metric space \Xdalfa. Suppose a point $x \in S$ has an open neighbourhood $N$ in $S$ which is homeomorphic to $\R$. Then there exist computable points $a, b \in N$ and a computable neighbourhood $N' \subseteq N$ of $x$ in $S$ which is an arc from $a$ to $b$. 
\end{thm}
\begin{proof}
    Let $N$ be an open neigborhood of $x$ in $S$ and let $f : \R \to N$ be a homeomorphism. Without loss of generality (we can always compose $f$ with a translation by $f^{-1}(x)$), we may assume $f(0)=x$. 

    Consider the subsets $f([-3, 3])$ and $f(\langle -4, 4 \rangle)$ of $S$. The set $f(\langle -4, 4\rangle)$ is a continuous image of an open set, so it is open in $N$ and therefore open in $S$. The set $f([-3, 3])$ is a nonempty compact subset of $f(\langle -4, 4 \rangle)$. By Corollary~\ref{cor:exists-S'}, there exists a semicomputable compact set $S'$ such that 
    \begin{equation}\label{eq:def-S'}
        f([-3, 3]) \subseteq S' \subseteq f(\langle -4, 4 \rangle) \text.
    \end{equation}

     Since $f^{-1}$ is continuous and $f([-4, 4])$ is compact, we can choose $\eps > 0$ such that 
    \begin{equation}\label{eq:def-eps}
    d\bigl(f(s), f(t)\bigr) < \eps \Longrightarrow \abs{s-t} < \frac12, \quad \text{for all } s, t \in [-4, 4]\text.
    \end{equation}
    We may additionally assume $\eps < 1$.

    The set $f(\langle -3, 3 \rangle)$ is a neighbourhood of $f(-2)$ in $S'$ homeomorphic to $\R$. By Theorem~\ref{tm:IV-5-2}, there is a computable neighbourhood $N_{\tilde a}$ of $f(-2)$ in $S'$. Similarly, there is a computable neighbourhood $N_{\tilde b}$ of $f(2)$ in $S'$.

    Since computable points are dense in computable sets,
    we can find computable points (see Figure~\ref{fig:a-b-kugle})
    \begin{equation}\label{eq:def-a-b-tilde}
        \tilde a \in B\bigl(f(-2), \eps\bigr) \cap N_{\tilde a} \qand 
        \tilde b \in B\bigl(f(2), \eps\bigr) \cap N_{\tilde b} \text.
    \end{equation}

    \begin{figure}[ht]
        \centering
             \begin{tikzpicture}[scale=0.92]
 \draw[
	orange!50!white, 
	line width = 3pt, 
	line cap=round, 
	smooth, 
	domain=0.5:3.5, 
	variable=\t] 
plot ({\t}, {sin(0.75*\t r)});
\draw[
	orange!50!white,
	line width = 3pt, 
	line cap=round, 
	smooth, 
	domain=6:10, 
	variable=\t] 
plot ({\t}, {sin(0.75*\t r)});
             
\draw[smooth, domain=0:10.75, variable=\t] plot ({\t}, {sin(0.75*\t r)});
\draw[dashed, smooth, domain=-0.75:0, variable=\t] plot ({\t}, {sin(0.75*\t r)});
\draw[dashed, smooth, domain=10.75:11.5, variable=\t] plot ({\t}, {sin(0.75*\t r)});

\draw[dotted,thick] (2,{sin(0.75*2 r)}) circle (1.3);
\draw[help lines,->] (2,{sin(0.75*2 r)}) to (2.92,{sin(0.75*2 r) + 0.92}) node [xshift=-6,yshift=-13] {$\eps$};

\draw[dotted,thick] (8,{sin(0.75*8 r)}) circle (1.3);
\draw[help lines,->] (8,{sin(0.75*8 r)}) to (8.92,{sin(0.75*8 r) - 0.92}) node [xshift=-6,yshift=13] {$\eps$};

\filldraw[fill=white] (1.25,{sin(0.75*1.25 r)}) node[label=$\tilde a$] {} circle (0.05);
\filldraw[fill=white] (2,{sin(0.75*2 r)}) node[label=$f(-2)$, yshift=-8mm] {}  circle (0.05);

\filldraw[fill=white] (7.25,{sin(0.75*7.25 r)}) node[label=$\tilde b$] {} circle (0.05);
\filldraw[fill=white] (8,{sin(0.75*8 r)}) node[label=$f(2)$, xshift=-1mm] {} circle (0.05);

\filldraw[fill=white] (4.5,{sin(0.75*4.5 r)}) node[label=$x$, xshift=1mm] {} circle (0.05);

\node[orange] at (3.75, 0.8) {$N_{\tilde a}$}; 
\node[orange] at (9.25, 1) {$N_{\tilde b}$}; 

     \end{tikzpicture}
        \caption{The choice of points $\tilde a$ and $\tilde b$ from $S'$.}
        \label{fig:a-b-kugle}
    \end{figure}

    Note that $\tilde a, \tilde b \in S' \subseteq f(\langle -4, 4\rangle)$, so $\tilde a = f(t_{\tilde a})$ and $\tilde b = f(t_{\tilde b})$ for some $t_{\tilde a}, t_{\tilde b} \in \langle -4, 4\rangle$. Moreover, by~\eqref{eq:def-eps}, we have
    \begin{equation}\label{eq:a-b-ograde}
        t_{\tilde a} \in \langle -2.5, -1.5\rangle \qand 
        t_{\tilde b} \in \langle 1.5, 2.5\rangle \text.
    \end{equation}

    For $(p, l, q) \in \N^3$, we consider the following statements:
    \begin{enumerate}[label=(O\arabic*)]
        \item\label{it:o1} $S' \subseteq J_p \cup J_\ugll \cup J_q$;
        \item\label{it:o2} $(J_p, J_{(l)_0}, \dots, J_{\zadnji l}, J_q)$ is a formal chain;
        \item\label{it:o3} $\tilde a \in J_p$;
        \item\label{it:o4} $\tilde b \in J_q$.
    \end{enumerate}

    Each of these relations is c.e.\ (which follows from, respectively, the semicomputability of $S'$, the computable enumerability of formal disjointness, the computability of $\tilde a$ and the computability of $\tilde b$), so the set
    \[ \Omega \defj \bigl\{l \in \N \bigm| (\exists p, q \in \N)\bigl(\text{\ref{it:o1}--\ref{it:o4} hold for } (p, l, q)\bigr) \bigr\} \]
    is also c.e.\ as a projection of their intersection.

    \begin{claim}\label{cl:1}
        There exists $l \in \Omega$ such that $\fmesh(l) < \frac\eps2$, $d\bigl(\tilde a, J_{(l)_0}\bigr) < \frac\eps2$ and $d\bigl(\tilde b, J_{\zadnji l}\bigr) < \frac\eps2$.
    \end{claim}
    \begin{proof}[Proof of Claim~\ref{cl:1}]\renewcommand{\qedsymbol}{}
        Since $f$ is uniformly continuous on $[-4, 4]$, we can find a subdivision
        $ -4 = x_0 < x_1 < \dots < x_{n+1}=4$
        of $[-4, 4]$ such that for each $i \in \{0, \dots, n\}$, 
        \begin{equation}\label{eq:subd-diam}
            \diam f([x_i, x_{i+1}]) < \frac\eps4 \text.
        \end{equation}
        Denote $C_i \defj f([x_i, x_{i+1}])$. 

        First, note that if $f(s) \in C_i$ and $f(t) \in C_j$, then 
        \begin{equation}\label{eq:ij-ts}
          \abs{j-i} > 2\abs{t-s} - 1 \text. 
        \end{equation}
        Indeed, suppose without loss of generality $s < t$ (the case $s=t$ is trivial); then $i \le j$ and
        \[t - s = ( t - x_j) + (x_j - x_{j-1}) + \dots + (x_{i+1} - s)\text,\]
        where we have $j-i+1$ summands. By~\eqref{eq:subd-diam} and~\eqref{eq:def-eps}, each one is less than $\frac12$, so $\abs{t-s} < \frac12(\abs{j-i}+1)$, and~\eqref{eq:ij-ts} follows from this. 
        
        Let $i$ be the index such that $t_{\tilde a} \in [x_i, x_{i+1} \rangle$ and let $j$ be the index such that $t_{\tilde b} \in \langle x_j, x_{j+1}]$. Then $\tilde a \in C_i$ and $\tilde b \in C_j$ and, by~\eqref{eq:a-b-ograde}, we have
        \begin{gather*}
        i = \abs{i-0} > 2\abs{t_{\tilde a}-(-4)} -1 \ge 2\abs{-2.5 + 4} - 1 = 2 \text,\\
        n - j = \abs{n-j} > 2\abs{4-t_{\tilde b}} -1 \ge 2 \abs{4-2.5} - 1 = 2  
        \shortintertext{and}
        \abs{j-i} > 2 \abs{t_{\tilde b}-t_{\tilde a}} - 1 \ge 2 \abs{1.5-(-1.5)} -1 = 5 \text.
        \end{gather*}

        Now we consider the compact quasi-chain
        \[ (C_0 \cup \dots \cup C_i, C_{i+1}, \dots, C_{j-1}, C_j \cup \dots \cup C_n) \text.\]
        By Lemma~\ref{lm:napuhavanje} (for $A = \emptyset$), we can find $p, l, q \in \N$ such that
        \begin{gather}
        C_0 \cup \dots \cup C_i \subseteq_{\frac\eps{16}} J_p,\quad
        C_j \cup \dots \cup C_n \subseteq_{\frac\eps{16}} J_q, \\
        C_{i+1} \subseteq_{\frac\eps{16}} J_{(l)_0},\quad\dots,\quad C_{j-1} \subseteq_{\frac\eps{16}} J_{\zadnji l}\label{eq:cl1-sadrz} \text.
        \end{gather}
        and such that $\bigl(J_p, J_{(l)_0}, \dots, J_{\zadnji l}, J_q\bigr)$ is a formal chain.

        Now, we have $\tilde a \in C_i \subseteq J_p$, $\tilde b \in C_j \subseteq J_q$ and
        \[
        S' \subseteq f([-4, 4]) = \bigcup_{i=0}^n C_i \subseteq J_p \cup J_\ugll \cup J_q \text.
        \]
        Therefore, $(p, l, q)$ satisfies~\ref{it:o1}--\ref{it:o4}, so $l \in \Omega$. Moreover, since $\diam C_k < \frac\eps4$ for all $i < k < j$,~\eqref{eq:cl1-sadrz} and Lemma~\ref{lm:r-diam} imply that for all $u\in\ugll$,
        \[
        \fdiam (u) < 4\cdot \frac\eps{16} + \frac\eps4 = \frac\eps2 \text,
        \]
         so $\fmesh(l) < \frac\eps2$. Finally, since $f(x_{i+1}) \in C_i \cap C_{i+1} \subseteq C_i \cap J_{(l)_0}$, we have
        \[d\bigl(\tilde a, J_{(l)_0}\bigr) \le d\bigl(\tilde a, f(x_{i+1})\bigr) \le \diam C_i < \frac\eps2 \]
        and, similarly, $d\bigl(\tilde b, J_{\zadnji l}\bigr) < \frac\eps2$. This concludes the proof of Claim~\ref{cl:1}.
    \end{proof}

    \begin{claim}\label{cl:2}
        If $l \in \Omega$, then 
        \[\left( J_{(l)_0} \cap S', \dots, J_{\zadnji l} \cap S' \right)\]
        is an open chain in $S'$.
    \end{claim}

    \begin{proof}[Proof of Claim~\ref{cl:2}]\renewcommand{\qedsymbol}{}
        Let $p, q \in \N$ be such that~\ref{it:o1}--\ref{it:o4} hold for $(p, l, q)$. We will prove that 
        \[\mathcal C \defj \bigl( J_p \cap S', J_{(l)_0} \cap S', \dots, J_{\zadnji l} \cap S', J_q \cap S' \bigr)\]
        is an open chain in $S'$. 

        The links of $\mathcal C$ are clearly open in $S'$. By~\ref{it:o2}, $\mathcal C$ is a quasi-chain. It is sufficient to prove that neighbouring links of $\mathcal C$ intersect.

       By~\eqref{eq:a-b-ograde},~\eqref{eq:def-S'} and~\ref{it:o1}, we have
       \begin{equation}
    f([t_{\tilde a}, t_{\tilde b}]) \subseteq f([-3, 3]) \subseteq S' = \bigcup \mathcal C
    \text.
    \end{equation}
       If $J_{(l)_{i-1}}\!\cap S'$ and $J_{(l)_i}\!\cap S'$ were disjoint for some $i>0$, we would have
       \begin{align*}
       f([t_{\tilde a}, t_{\tilde b}]) &\subseteq \left((J_p \cap S') \cup \dots \cup (J_{(l)_{i-1}}\!\cap S') \right) \cup{}\\
       &{}\cup \left( (J_{(l)_i}\!\cap S') \cup \dots \cup (J_q \cap S') \right) \text.
       \end{align*}
       However, this is impossible: because $f(t_{\tilde a}) \in J_p \cap S'$ and $f(t_{\tilde b}) \in J_q \cap S'$ by~\ref{it:o3} and~\ref{it:o4}, the right-hand side is a disjoint union of nonempty open sets, but $f([t_{\tilde a}, t_{\tilde b}])$ is connected. The same argument shows that $J_p \cap S'$ intersects $J_{(l)_0} \cap S'$ and $J_{\zadnji l} \cap S'$ intersects $J_q \cap S'$.  
       This concludes the proof of Claim~\ref{cl:2}.
    \end{proof}

    For $(l, l') \in \N^2$, we consider the following statements:
    \begin{enumerate}[label=(G\arabic*)]
        \item\label{it:g1} $l, l' \in \Omega$;
        \item\label{it:g2} $J_\ugllc\sA J_\ugll$;
        \item\label{it:g3} $\fmesh(l') < \frac12\fmesh(l)$;
        \item\label{it:g4} $J_{(l')_0} \sA J_{(l)_0}$;
        \item\label{it:g5} $J_{\zadnji{l'}} \sA J_{\zadnji l}$.
    \end{enumerate}

    Denote by $\Gamma$ the set of all $(l, l') \in \N^2$ which satisfy~\ref{it:g1}--\ref{it:g5}. By computable enumerability of $\Omega$ and $\sA$ (as a relation on $J$-indices) and computability of $\fmesh$, $\Gamma$ is c.e.

    \begin{claim}\label{cl:3}
        For each $l \in \Omega$, there exists $l' \in \Omega$ such that $(l, l') \in \Gamma$.
    \end{claim}
    \begin{proof}[Proof of Claim~\ref{cl:3}]\renewcommand{\qedsymbol}{}
        Let $l \in \Omega$ and let $p, q \in \N$ be such that~\ref{it:o1}--\ref{it:o4} hold for $(p, l, q)$.

        By~\ref{it:o1}, $\{J_p, J_{(l)_0}, \dots, J_{\zadnji l}, J_q\}$ is an open cover of $S'$. Let $\lambda$ be a Lebesgue number of this cover.

        Since $J_p$ and $J_q$ are open and $\tilde a \in J_p$, $\tilde b \in J_q$, there exist $r_1,r_2 > 0$ such that $B(\tilde a, r_1) \subseteq J_p$ and $B(\tilde b, r_2) \subseteq J_q$. Denote
        \[ r \defj \min\, \bigl\{\lambda, r_1,r_2, \textstyle \frac{1}{10}\fmesh(l) \bigr\} \text.\]

        Let $ -4 = x_0 < x_1 < \dots < x_{n+1}=4$ be a subdivision
        of $[-4, 4]$ such that for each $i \in \{0, \dots, n\}$
        \begin{equation}\label{eq:subd-diam-2}
        \diam f([x_i, x_{i+1}]) < r \text.
        \end{equation}
        Let $C_i \defj f([x_i, x_{i+1}]) \cap S',\forall i\in\{0,\dots,n\}$. Then $(C_0, \dots, C_n)$ is a compact quasi-chain in~$S'$ which covers~$S'$. 

        As in the proof of Claim~\ref{cl:1}, let $i$ be the index such that $t_{\tilde a} \in [x_i, x_{i+1} \rangle$ and let $j$ be the index such that $t_{\tilde b} \in \langle x_j, x_{j+1}]$, so that $\tilde a \in C_i$ and $\tilde b \in C_j$. As before, we know that $i > 2$, $j < n-2$ and $\abs{i-j} > 5$. 

        Since $\tilde a \in C_i$, $\diam C_i < r_1$ and $B(\tilde a, r_1) \subseteq J_p$ we have $C_i \subseteq J_p$. Similarly, $C_j \subseteq J_q$. By Lemma~\ref{lm:links}, there exists $t$ such that $i < t < j$ and $C_t \subseteq J_{(l)_0}$.
        Let 
        \[ v \defj \max\,\bigl\{t \in \{i+1, \dots, j-1\} \bigm| C_t \subseteq J_{(l)_0} \bigr\} \text. \]

        Again, by Lemma~\ref{lm:links}, there exists $s$ such that $v \le s < j$ and $C_s \subseteq J_{\zadnji l}$. Let
        \[ w \defj \min\,\bigl\{s \in \{v, \dots, j-1\} \bigm| C_s \subseteq J_{\zadnji l} \bigr\} \text. \]

        We claim that $(C_v, \dots C_w)$ strongly refines $\left(J_{(l)_0}, \dots, J_{\zadnji l} \right)$. By definition, $C_v \subseteq J_{(l)_0}$ and $C_w \subseteq J_{\zadnji l}$. Suppose $k \in \{v+1, \dots, w-1\}$. Then, since $\diam C_k < \lambda$, $C_k$ must be a subset of $J_p$, $J_q$ or $J_i$ for some $i \in\ugll$. Suppose $C_k \subseteq J_p$. Since $C_w \subseteq J_{\zadnji l}$, $J_p$ strictly precedes $J_{(l)_0}$ and $J_{(l)_0}$ strictly precedes $J_{\zadnji l}$, by Lemma~\ref{lm:links} there exists $u$ such that $k < u < w$ and $C_u \subseteq J_{(l)_0}$. However, this contradicts our choice of $v$ because $v < k < u$. A similar argument shows that $C_k \not\subseteq J_q$, so $C_k$ is a subset of $J_i$ for some $i\in\ugll$.

        To summarize, 
        \[
        (C_0 \cup \dots \cup C_{v-1}, C_v, \dots, C_w, C_{w+1} \cup \dots \cup C_n)
        \]
        is a compact quasi-chain in $S'$ which covers $S'$ and its subchain $(C_v, \dots C_w)$ strongly refines $\bigl(J_{(l)_0}, \dots, J_{\zadnji l} \bigr)$.
        By Lemma~\ref{lm:napuhavanje} (with $A=[l]$), we can find $p', l', q' \in \N$ such that
        \begin{gather}
        C_0 \cup \dots \cup C_{v-1} \subseteq_{r} J_{p'},\quad
        C_{w+1} \cup \dots \cup C_n \subseteq_{r} J_{q'} \text, \\
        C_{v} \subseteq_{r} J_{(l')_0},\quad\dots,\quad C_{w} \subseteq_{r} J_{\zadnji{l'}} \text,\label{eq:gamma-eps} \\
        J_{(l')_0}\sA J_{(l)_0}, \quad J_{\zadnji{l'}} \sA J_{\zadnji l} \qand J_\ugllc\sA J_\ugll\label{eq:gamma-ref}
        \end{gather}
        and such that $\bigl(J_{p'}, J_{(l')_0}, \dots, J_{\zadnji{l'}}, J_{q'}\bigr)$ is a formal chain.

        Moreover, we have $\tilde a \in C_i \subseteq J_{p'}$, $\tilde b \in C_j \subseteq J_{q'}$ and
        \[ S' = f([-4, 4]) \cap S' = \bigcup\nolimits_{i=1}^{\;n} C_i \subseteq J_{p'} \cup J_\ugllc\cup J_{q'} \text.\]
        Therefore, $(p', l', q')$ satisfies~\ref{it:o1}--\ref{it:o4}, which implies~\ref{it:g1}; while~\ref{it:g2},~\ref{it:g4} and~\ref{it:g5} follow from~\eqref{eq:gamma-ref}. Finally,~\eqref{eq:subd-diam-2},~\eqref{eq:gamma-eps} and Lemma~\ref{lm:r-diam} imply that 
        \[
        \fdiam (k) < 4r + r \le \frac12 \fmesh (l)\text{,\quad for each $k \in \ugll$,}
        \]
        so~\ref{it:g3} also holds. This proves $(l, l') \in \Gamma$, so Claim~\ref{cl:3} holds.
    \end{proof}

    Claim~\ref{cl:3} implies the existence of a partial recursive funcion $\psi : \Omega \to \N$ such that $\bigl(l, \psi(l)\bigr) \in \Gamma$ for all $l \in \Omega$. Let $l_0 \in \Omega$ be as in Claim~\ref{cl:1} and define the sequence $(l_n)_{n \in \N}$ of natural numbers with
    \[
        l_{n+1} \defj \psi(l_n), \qquad \forall n \in \N \text.
    \]
    Obviously, $(l_n)_{n \in \N}$ is recursive and $(l_n, l_{n+1}) \in \Gamma$ for all $n \in \N$.

    Consider the sequence $(\mathcal C_n)_{n\in\N}$, where
    \begin{equation}
        \mathcal C_n \defj \Bigl(J_{(l_n)_0} \cap S', \dots, J_{\zadnji{l_n}} \cap S'\Bigr), \qquad \forall n \in \N \text.
    \end{equation}
    By Claim~\ref{cl:2}, each $\mathcal C_n$ is an open chain in $S'$. As a subspace of $(X, d)$, $S'$ is a metric space in which every closed ball is compact (since $S'$ itself is compact). Also,~\ref{it:g2} and Proposition~\ref{prop:profinjenje}\ref{it:prof1} imply that 
    \begin{equation} \label{eq:prof}
        \Bigl( \cl\bigl(J_{(l_{n+1})_0} \cap S'\bigr), \dots, \cl\bigl(J_\zadnji{l_{n+1}} \cap S'\bigr) \Bigr) \ \text{strongly refines}\ \mathcal C_n
    \end{equation}
    for each $n \in \N$. It follows from Claim~\ref{cl:1} and~\ref{it:g3} that 
    \begin{equation} \label{mash-brzo-pada}
        \mesh (\mathcal C_n) < 2^{-n}, \qquad \forall n \in \N\text.
    \end{equation}

    By Lemma~\ref{lm:presjek-lanaca},
    \begin{equation} \label{def-N'}
        N' \defj \bigcap\nolimits_{n \in \N}\bigl(\,\bigcup\nolimits_{j\in[\mspace{2mu}l_{n+1}\mspace{1mu}]} \cl(J_j \cap S')\bigr)
    \end{equation}
    is a continuum  in $S'$ chainable from $a$ to $b$, where 
    \begin{equation} \label{def-a-b}
        a \in \bigcap\nolimits_{n \in \N} \bigl(J_{(l_n)_0} \cap S'\bigr) \qand b \in \bigcap\nolimits_{n \in \N} \bigl(J_{\zadnji{l_n}} \cap S'\bigr) \text.
    \end{equation}

    Since 
    \[
    N' \subseteq S' \subseteq f(\langle -4, 4 \rangle) \subseteq N \text,
    \]
    $N'$ is a subset of $N$ and a compact connected subset of $f(\langle -4, 4 \rangle)$, so it must be an arc. Since an arc can only be chainable from one of its endpoints to the other, $a$ and $b$ must be endpoints of $N'$.
    
   Since $\fmesh(l_n) < 2^{-n}$ for each $n \in \N$, \eqref{def-a-b} implies
   \[d\bigl(a, \lambda_{((l_n)_0)_0}\bigr) < 2^{-n} \qand d\bigl(b, \lambda_{(\zadnji{l_n})_0}\bigr) < 2^{-n}\]
   for each $n \in \N$. It follows that $a$ and $b$ are computable points.

   The function
   $ n \mapsto \left\lbrace \lambda_{(i)_0} \bigm| i \in \uglln \right\rbrace$ is computable. Let $n\in\N$. We claim that 
   \begin{equation} \label{efektivna-aproksimacija}
    N' \approx_{2^{-n}}\! \left\lbrace \lambda_{(i)_0} \bigm| i \in \uglln \right\rbrace\text.
   \end{equation}

   Let $y \in N'$. By~\eqref{def-N'} and~\eqref{eq:prof}, $y \in J_i \cap S'$ for some $i \in \uglln$. Since $\fmesh(l_n) < 2^{-n}$ and therefore $\fdiam(i) < 2^{-n}$, it holds that
   \[ \diam (J_i \cap S') \le \diam (J_i) \le \fdiam(i) < 2^{-n} \text.\]
    Therefore, $d\bigl(y, \lambda_{(i)_0}\bigr) < 2^{-n}$.

    Let $i \in \uglln$. Since $\mathcal C_n$ is a chain which covers $N'$, its first link contains $a \in N'$ and its last link contains $b \in N'$, the connectedness of $N'$ implies that its link $J_i \cap S'$ must intersect $N'$. Let $y \in N' \cap (J_i \cap S')$. Again, using $\fdiam(i) < 2^{-n}$ we conclude that
    $d\bigl(y, \lambda_{(i)_0}\bigr) < 2^{-n}$. 
    
    It is not hard to conclude~\citeup{Proposition 2.6}{lmcs:mnf} that there exists a computable function $\phi : \N \to \N$ such that
    \[\left\lbrace \lambda_{(i)_0} \bigm| i \in \uglln \right\rbrace = \{\alpha_i \mid i \in [\phi(k)] \} \text,\]
    so~\eqref{efektivna-aproksimacija} implies that $N'$ is a computable set in $(X, d, \alpha)$. This concludes the proof of Theorem~\ref{tm:compneigh}.
\end{proof}

Theorem~\ref{tm:compneigh} states that, if a point in a semicomputable set has a neighbourhood homeomorphic to $\R$, then it has a computable neighbourhood which is an arc with computable endpoints. A natural potential strengthening of this statement would be to require the computable neighbourhood to be \textit{computably} homeomorphic to the unit segment $[0, 1]$. However, this stronger statement does not hold, as demonstrated by the following example.

\begin{exa}
    Let $B \subseteq \N$ be a non-computable c.e.\ set and let $(b_i)_{i \in \N}$ be an injective computable enumeration of $B$. For each $n \in \N$, let
    \[ T_n \defj \left[ \frac{1}{2^{n+1}}, \frac{1}{2^n} \right] \times \{0\} \subseteq \R^2 \]
    and for $n = b_i \in B$ let $S_n$ be the polygonal arc in $\R^2$ connecting
    \[ \left( \frac{1}{2^{n+1}}, 0\right)\text, 
    \left( \frac{5}{3}\cdot\frac{1}{2^{n+1}}, \frac{1}{2^{i+1}}\right)\text, 
    \left( \frac{4}{3}\cdot\frac{1}{2^{n+1}}, -\frac{1}{2^{i+1}}\right)\text{ and } 
    \left( \frac{1}{2^n}, 0 \right) \]
    (see Figure~\ref{fig:TnSn}.)

    \begin{figure}[ht]
        \centering
        \begin{tikzpicture}
            \draw[thick] (0, 0) -- (3, 0);
            \draw[thick, gray!50!white] (6, 0) -- (9, 0);
            \draw[thick] (6, 0) -- (8, 0.66) -- (7, -0.66) -- (9, 0);

            \fill (0, 0) circle (0.05);
            \fill (3, 0) circle (0.05);
            \fill (6, 0) circle (0.05);
            \fill (9, 0) circle (0.05);

            \node[anchor=north] at (0, 0) {$\frac{1}{2^{n+1}}$};
            \node[anchor=north] at (3, 0) {$\frac{1}{2^n}$};
            \node[anchor=north] at (6, 0) {$\frac{1}{2^{n+1}}$};
            \node[anchor=north] at (9, 0) {$\frac{1}{2^n}$};

            \node at (1.5, -1) {$T_n$};
            \node at (7.5, -1) {$S_n$};
        \end{tikzpicture}
        \caption{The sets $T_n$ and $S_n$ for $n = b_i \in B$}
        \label{fig:TnSn}
    \end{figure}

    Let
    \[ V_n \defj \begin{cases}
        T_n, \quad& n \notin B \\
        S_n, \quad& n \in B 
    \end{cases}\]
    and let
    \[ A := \bigcup_{n \in \N} V_n \cup [-1, 0] \times \{0\} \subseteq \R^2 \text.\]

    Then $A$ is a computable arc in $\R^2$ which is not an image of the unit segment $[0, 1]$ under any computable injection (see \cite[Example 5.1]{miller}). 

    The point $(0, 0)$ has a neighbourhood in $A$ homeomorphic to $\R$, so by Theorem~\ref{tm:compneigh} it has a neighbourhood which is a computable arc with computable endpoints (e.g.\ $A$ itself). However, there is no neighbourhood of $(0, 0)$ in $A$ \textit{computably} homeomorphic to $[0, 1]$. 

    Suppose $U$ is a neighbourhood of $(0, 0)$ in $A$ and $f : [0, 1] \to U$ is a computable homeomorphism. Without loss of generality, we may assume $f(0) = (s, 0)$ and $f(1)=(t, u)$, where $s < 0 < t$. Note that both $(s, 0)$ and $(t, u)$ are computable points. Moreover, note that the subarc $A'$ of $A$ between $(t, u)$ and $(1, 0)$ is a finite union of line segments with rational endpoints, so it has a computable parametrization $g : [0, 1] \to A'$ such that $g(0)=(t, u)$ and $g(1)=(1, 0)$. Now, the function $h : [0, 1] \to A$,
    \[ h(x) = \begin{cases}
        3x\cdot(s+1)-1, \quad & x \in [0, \frac{1}{3}] \\
        f(3x-1), & x \in [\frac{1}{3}, \frac{2}{3}] \\
        g(3x-2), & x \in [\frac{2}{3}, 1] 
    \end{cases} \]
    is a computable homeomorphism between $[0, 1]$ and $A$. Since no such homeomorphism (in fact, no such injection) exists, we conclude that no neighbourhood of $(0, 0)$ in $A$ is computably homeomorphic to $[0, 1]$.
\end{exa}

\begin{lem}\label{lm:izrezivanje}
    Let $S$ be a semicomputable set in a computable metric space \Xdalfa. Suppose a point $x \in S$ has an open neighbourhood $N$ in $S$ such that there exists a homeomorphism $f : [0, 1\rangle \to N$ such that $f(0)=x$. Let $\eps > 0$. Then there exists $a \in \langle 0, 1 \rangle$ such that $f(a)$ is a computable point, $S\setminus f([0, a\rangle)$ is a semicomputable set and $f([0, a]) \subseteq B(x, \eps)$.
\end{lem}
\begin{proof}
    Since $f$ is continuous, there exists $t\in \langle 0,1\rangle$ such that
\begin{equation}\label{eq:0-t-u-kugli}
f([0,t])\subseteq B(x,\eps) \text.
\end{equation}
The set $f(\langle 0, 1 \rangle)$ is open in $N$ and therefore open in $S$. 
Therefore $f(\langle 0, 1 \rangle)$ is an open neighbourhood of $f(t)$ in $S$ which is homeomorphic to $\R$ and so Theorem~\ref{tm:compneigh} implies that there exists a computable neighbourhood $N'$ of $f(t)$ in $S$ contained in $f(\langle 0, 1 \rangle)$ which is an arc with computable endpoints. Since $f$ is a homeomorphism, $f^{-1}(N')$ is an arc in $\langle 0,1\rangle$, thus it is equal to $[a,b]$ for some $a,b\in \langle 0,1 \rangle $, $a < t < b$ (see Figure~\ref{fig:izrezivanje}).

\begin{figure}[ht]
    \centering
   \begin{tikzpicture}
\draw[orange!50!white, line width = 1.5mm, smooth, domain=0.6:2.75, variable=\t] plot ({\t}, {0.6*sin(\t r)});

\draw[thick, smooth, domain=0:6.282, variable=\t] plot ({\t}, {0.6*sin(\t r)});
\draw[thick, smooth, dashed, domain=6.282:7, variable=\t] plot ({\t}, {0.6*sin(\t r)});

\draw[dashed, gray] (0, 2.25) arc (90:-30:2.25);
\draw[->, gray] (0, 0) -- (2.2, -0.5);

\filldraw[fill=white] (0, 0) circle (0.065); 
\node[anchor=north] at (0, 0) {$x$};
\node[anchor=south east] at (0, 0) {$f(0)$};

\filldraw[fill=white] (1.5, {0.6*sin(1.5 r)}) node[label=$f(t)$] {} circle (0.065);
\filldraw[fill=white] (0.6, {0.6*sin(0.6 r)}) node[label=$f(a)$] {} circle (0.065);
\filldraw[fill=white] (2.75, {0.6*sin(2.75 r)}) node[label=$f(b)$] {} circle (0.065);

\node[anchor=north, gray] at (1, -0.25) {$\eps$};

\end{tikzpicture}
    \caption{Points along an arc as in the proof of Lemma~\ref{lm:izrezivanje}.\\ The highlighted arc is computable.}
    \label{fig:izrezivanje}
\end{figure}

Clearly, $f(a)$ is a computable point and 
\[f([0, a]) \subseteq f([0, t]) \subseteq B(x, \eps) \text.\] We claim that $S \setminus f([0, a \rangle)$ is a semicomputable set.

The set $f([0, b\rangle)$ is open in $N$ and therefore open in $S$. Let $U$ be an open set in \Xd\ such that $f([0, b\rangle) = S \cap U$. The set $U$ is open, so it is a union of rational open balls. Since $f([0, a])$ is a compact set contained in $U$, there is a finite union of rational open balls $J_m$ such that 
\begin{equation}\label{eq:Jm}
    f([0, a]) \subseteq J_m \subseteq U \text.
\end{equation}

By Lemma~\ref{lm:S-minus-J-semicomp}, $S \setminus J_m$ is semicomputable. Since $f([a, b])$ is computable, it is also semicomputable, so $(S \setminus J_m) \cup f([a, b])$ is semicomputable. 

Now $f(a)$ is a computable point and $f([0, a]) \subseteq B(x, \eps)$. To prove that $S \setminus f([0, a \rangle)$ is semicomputable, we show that it is equal to $(S \setminus J_m) \cup f([a, b])$.

Indeed, since $f$ is injective, $f([a, b])$ and $f[0, a\rangle)$ are disjoint, so $f([a, b]) \subseteq S \setminus f([0, a\rangle)$. It follows from~\eqref{eq:Jm} that $S \setminus J_m \subseteq S \setminus f([0, a\rangle)$. Therefore, $(S \setminus J_m) \cup f([a, b]) \subseteq S \setminus f([0, a \rangle)$.

Suppose $y \in S \setminus f([0, a\rangle)$. If $y \in J_m$, then by \eqref{eq:Jm} 
\[y \in U \cap S = f([0, b\rangle) = f([0, a\rangle) \cup f([a, b\rangle) \text,\]
which implies $y \in f([a, b\rangle)$. This proves $S \setminus f([0, a \rangle) \subseteq (S \setminus J_m) \cup f([a, b])$.
\end{proof}

\section{Computable approximations of semicomputable sets}

In this section, we consider \emph{(topological) graphs} \cite{mlq-graphs}, i.e.,\ disjoint unions of arcs and rays that can be glued together at their endpoints. We will show that every semicomputable graph can be approximated by a computable subgraph with computable endpoints. This subgraph differs from the original graph only near the uncomputable endpoints of the graph, where computable ones are found arbitrarily close and the corresponding edges made that much smaller.

The main idea is the following: suppose we are given a semicomputable graph with possibly uncomputable endpoints. Lemma~\ref{lm:izrezivanje} allows us to shorten an edge of the graph ending in an uncomputable point so that it ends in a computable point. We then repeat this procedure to get a semicomputable subgraph with computable endpoints. To complete the proof that the resulting graph is indeed computable, we leverage a result from the literature~\cite[Theorem 5.2]{mlq-graphs} which states that a topological pair of a graph and the set of all its endpoints has computable type. In particular, a semicomputable graph is computable if the set of all its endpoints is semicomputable; this generalizes the result~\cite[Theorem~7.5]{lmcs:1mnf} that each semicomputable 1-manifold with
semicomputable boundary and finitely many components is computable.

\subsection{Compact graphs}

Let $G$ be a compact graph. Then $G = \bigcup K$, where
$K$ is a nonempty finite family of (non-degenerate) line segments in $\R^n$ such that the following holds:
\begin{equation}\label{eq:simplcompl}
\text{if }I,J\in K\text{ are such that }I\ne J\text{ and }I\cap J\ne\emptyset\text{, then }I\cap  J=\{a\},
\end{equation}
where $a$ is an endpoint of both $I$ and $J$.

Equivalently, a topological space $G$ is a compact graph if and only if there exists a nonempty finite family $\mathcal A$ of subspaces of $G$ such that each $A\in \mathcal A$ is an arc, such that $G=\bigcup_{A\in \mathcal A}A$ and such that for all $A,B\in \mathcal A$ the following holds:
\begin{equation}\label{eq:graph-arcs}
\text{if }A,B\in \mathcal A\text{ are such that }A\ne B\text{ and }A\cap B\ne\emptyset \text{, then }A\cap  B=\{a\},
\end{equation}
where $a$ is an endpoint of both $A$ and $B$~\cite[Remark 4.2]{mlq-graphs}. If $G$ is a compact graph and $\mathcal A$ is a family with the described properties, then we say that $\mathcal A$ \defn{defines} $G$.

If $G$ is a compact graph and if $\mathcal A$ defines $G$, then $x\in G$ is an endpoint of $G$ if and only if there exists a unique arc $A\in\mathcal A$ such that $x$ is an endpoint of $A$~\cite{mlq-graphs}. In particular, the set of all endpoints of a compact graph $G$ is finite.

If $G$ is a compact graph and $E$ the set of all its endpoints, then $(G,E)$ has computable type~\cite{mlq-graphs}. Hence, if \Xdalfa\ is a computable metric space and $S$ a semicomputable set in this space which is, as a subspace of \Xd, a compact graph, then $S$ is computable if each endpoint of $S$ is computable.

Suppose $S$ is a semicomputable compact graph in a computable metric space \Xdalfa. If not all of its endpoints are computable, then $S$ need not be computable (even a semicomputable arc in $\R$ need not be computable~\cite{miller}). We want to show that such an $S$ can be approximated by a computable subgraph $T$. Namely, for each uncomputable endpoint $a$ of $S$, we cut off some small neighbourhood of $a$ in $S$ and get a computable graph $T$ with computable endpoints (see Figure~\ref{fig:graph}). In general, endpoints of a computable graph $T$ need not be computable, even when $T$ is an arc~\cite{miller}.

\begin{figure}[ht]
    \centering
    \begin{tikzpicture}
\draw[line width=1.5mm, line cap=round, orange!50!white] (1.2,6) -- (2,6) -- (3,5)
			-- (3.5, 6.5) -- (2, 6);
\draw[line width=1.5mm, line cap=round, orange!50!white] (2, 8) -- (3.5, 6.5) -- (5,7);
\draw[line width=1.5mm, line cap=round, orange!50!white] (3, 5) -- (4.6, 4.2);
\draw (1,6) -- (2,6) -- (3,5);
\draw (0.94,6) circle (0.06);
\draw (3,5) -- (4.94,4.03);
\draw (2,6) -- (3.5,6.5) -- (3,5);
\draw (2,8) -- (3.5,6.5);
\draw (3.5,6.5) -- (5,7);
\fill (5,7) circle (0.07);
\draw[gray, dashed] (0.94,5.6) arc (-60:60:0.5);
\fill (2, 8) circle (0.07);
\draw (5, 4) circle (0.07);
\draw[gray, dashed] (4.5, 3.8) arc (180:80:0.6);
\end{tikzpicture}
    \caption{A semicomputable compact graph. Empty circles denote uncomputable points and filled circles denote computable points. The highlighted subset is a computable compact graph.}
    \label{fig:graph}
\end{figure}

For simplicity, we will use the following notation: if $A$ is an arc and $x,y\in A$ are such that $x\ne y$, then the unique subspace of $A$ which is an arc whose endpoints are $x$ and $y$ will be denoted by $\luk xyA$.

\begin{thm}\label{tm:zadnjidio}
Let \Xdalfa\ be a computable metric space and let $S$ be a semicomputable compact graph in this space. Let $\{A_0,\dots,A_n\}$ be a family of arcs which defines $S$.

Let $\Gamma $ be the set of all endpoints of $S$ which are uncomputable. Let $\eps >0$.

Then there exist arcs $A_0',\dots,A_n'$ such that for each $i\in \{0,\dots,n\}$ the following hold:
\begin{enumerate}[label=(\roman*)]
    \item if none of the endpoints of the arc $A_i $ belong to $\Gamma $, then $A_i'=A_i $;
    \item if $x$ and $y$ are endpoints of $A_i $ such that $x\in \Gamma $ and $y\notin \Gamma $, then $A_i'=\luk zy{A_i}$, where $z\in A_i $ and $\luk xz{A_i}\subseteq B(x,\eps) $;
    \item if $x$ and $y$ are endpoints of $A_i $ such that $x,y\in \Gamma $,  then $A_i'=\luk{z_1}{z_2}{A_i}$, where $z_1,z_2 \in A_i $ and $\luk x{z_1}{A_i}\subseteq B(x,\eps)$,  $\luk{z_2}y{A_i}\subseteq B(y,\eps)$.
    \item the set $T=A_0'\cup \dots \cup A_n'$ is a computable compact graph with computable endpoints.
\end{enumerate}
\end{thm}
\begin{proof}
Let us define $A_0'$ in the following way. If none of the endpoints of $A_0 $ belong to $\Gamma $, let $A_0'=A_0$.

Suppose that $x\in \Gamma $ and $y\notin \Gamma $, where $x$ and $y$ are the endpoints of $A_0 $. There exists a homeomorphism $f:[0,1]\to A_0 $ such that $f(0)=x$ and $f(1)=y$. The set $f([0, 1 \rangle)$ is open in $S$: its complement in $S$ is equal to (note that $A_j \cap f( \langle 0, 1 \rangle) = \emptyset$ for $j \ne 0$ since $A_j$ can intersect $A_0$ only in an endpoint, by the definition of a graph)
\[ \{ f(1)\} \cup \bigcup\nolimits_{i \ne 0} A_i\text, \]
which is clearly a closed set in $S$.

Therefore, $f([0, 1\rangle)$ is an open neighbourhood of $x$ in $S$. By Lemma~\ref{lm:izrezivanje}, there exists $a \in \langle 0, 1\rangle $ such that $f(a)$ is a computable point, $S \setminus f([0, a \rangle)$ is a semicomputable set and $f([0, a]) \subseteq B(x, \eps)$. 

Let $z = f(a) \in A_0$ and $A_0' = \luk z y {A_0}$. It holds that $\luk x z{A_0} \subseteq B(x, \eps)$. Now, the family of arcs $\{A_0',A_1,\dots,A_n\}$ defines a compact graph $T$ such that the set of all uncomputable endpoints of $T$ is equal to $\Gamma \setminus \{x\}$. We also have
\[ T = S \setminus f([0, a \rangle) \text,\]
so $T$ is semicomputable.

Finally, if both endpoints $x$ and $y$ of $A_0 $ belong to $\Gamma $, we proceed as in the previous case to obtain an arc of the form $\luk x{z_2}{A_0} $, where $z_2 $ is a computable point such that $z_2 \in A_0$, $\luk{z_2}{y}{A_0}\subseteq B(y,\eps)$ and such that $\luk{x}{z_2}{A_0}\cup \bigcup_{i\ne 0}A_i$ is a semicomputable set. We then apply a similar procedure to the point $x$ to obtain an arc of the form $\luk{z_1}{z_2}{A_0}$, where $z_1 $ is a computable point such that $z_1 \in \luk{x}{z_2}{A_0} $, $\luk{x}{z_1}{A_0}\subseteq B(x,\eps)$ and such that $\luk{z_1}{z_2}{A_0}\cup \bigcup_{i\ne 0}A_i$ is a semicomputable set. We define $A_0'=\luk{z_1}{z_2}{A_0}$.

As in the previous case, we have that the family $\{A_0',A_1,\dots,A_n\}$ defines a semicomputable compact graph $T$ such that the set of all uncomputable endpoints of $T$ is equal to $\Gamma\setminus\{x,y\}$.

Now we define $A_1'$ in the same way and we get that the family $\{A_0',A_1',A_2,$ $\dots,A_n\}$ defines a semicomputable compact graph $T$ such that the set of all uncomputable endpoints of $T$ is equal to
$\{x\in \Gamma \mid  x$ is an endpoint of $A_i $ for some $i\ge 2\}$.

In finitely many steps we get $A_0',\dots,A_n'$ which satisfy properties (i)--(iii) such that $\{A_0',\dots,A_n'\}$ defines a semicomputable compact graph $T$ which has no uncomputable endpoints. Hence, every endpoint of $T$ is computable, so $T$ is computable by Theorem~\ref{tm:grafovi}.
\end{proof}

An immediate consequence of Theorem~\ref{tm:zadnjidio} is the following fact.
\begin{cor}
Let \Xdalfa\ be a computable metric space and let $S$ be a semicomputable compact graph in this space. Then for each $\eps >0$ there exists a computable compact graph $T$ in \Xdalfa\ such that $T\subseteq S$, all endpoints of $T$ are computable and $d_{H}(S,T)<\eps $.
\end{cor}

\subsection{Non-compact graphs}

Having generalized notions of a semicomputable and computable set to non-compact sets, the question is under what conditions the implication
\begin{equation}\label{eq:semi-comp-noncomp}
S\text{ semicomputable}\Longrightarrow S\text{ computable}
\end{equation}
holds. Beside the known conditions when $S$ is compact, it is known that~\eqref{eq:semi-comp-noncomp} holds when $S$ is a graph such that each endpoint of $S$ is computable~\cite{mlq-graphs}.

Any topological space homeomorphic to a ray is called a \defn{topological ray}. If $R$ is a topological ray and $a\in R$ such that $R\setminus \{a\}$ is connected, then we say that $a$ is an \defn{endpoint} of $R$. Note that each topological ray has a unique endpoint.

If $G$ is a graph, then it follows easily from the definition that there exists a finite nonempty family $\mathcal A$ of closed subspaces of $G$ such that each element of $\mathcal A$ is an arc or a topological ray, $G=\bigcup \mathcal A$ and for all $A,B\in \mathcal A$ the following holds:
\begin{equation}\label{eq:gengraph-arcs}
\text{if }A,B\in \mathcal A\text{ are such that }A\ne B\text{ and }A\cap B\ne\emptyset \text{, then }A\cap  B=\{a\},
\end{equation}
where $a$ is an endpoint of both $A$ and $B$.

Conversely, let $G$ be a topological space such that $G=\bigcup \mathcal A$, where $\mathcal A$ is a finite nonempty family of closed subspaces of $G$ such that each element of $\mathcal A$ is an arc or a topological ray and such that for all $A,B\in \mathcal A$ the implication~\eqref{eq:gengraph-arcs} holds. We want to show that $G$ is a graph.

Suppose $\mathcal A=\{A_0,\dots,A_m\}$. Similarly to~\cite[Remark 4.2]{mlq-graphs}, we get that there exist $n\in \N\setminus \{0\}$ and finitely many subsets $F_0,\dots,F_m $ of $\R^n$ such that $K_i $ is a line segment or a ray for each $i\in \{0,\dots,m\}$ and for each  $i\in \{0,\dots,m\}$ there exists a homeomorphism $h_i:A_i\to K_i $ such that for all $i\ne j$ the following holds:
\begin{equation}\label{eq:zadnjidio-h-ovi}
h_i (A_i \cap A_j) =h_j (A_i \cap A_j)=h_i (A_i)\cap h_j (A_j).    
\end{equation}
The fact that $A_0,\dots,A_m $ are closed in $G$ allows us to glue the maps $h_0,\dots,h_m $ and to get a continuous bijection $f:G\to K_0 \cup \dots \cup K_m$. Each of the sets $K_0,\dots,K_m$ is closed in $\R^n$, thus also in $K_0 \cup \dots \cup K_m$, and so $f^{-1} $ is continuous too. Hence $f$ is a homeomorphism. The equalities~\eqref{eq:zadnjidio-h-ovi} imply that the family $\{K_0,\dots,K_m\}$ satisfies~\eqref{eq:simplcompl-2}, which means that $G$ is a graph.

If $G$ is a  graph and $\mathcal A$ is a finite nonempty family of closed subspaces of $G$ such that each element of $\mathcal A$ is an arc or a topological ray and for all $A,B\in \mathcal A$ the implication~\eqref{eq:gengraph-arcs} holds, then we say that $\mathcal A$ \defn{defines} $G$.

Let $R$ be a topological ray and let $x,y\in R$, $x\ne y$. By $\luk xyR$ we denote a (unique) subspace of $R$ which is an arc with endpoints $x$ and $y$. Note the following: if $a$ is an endpoint of $R$ and $x\in R$, $x\ne a$, then
$(R\setminus\luk axR)\cup \{x\}$ is a topological ray and a closed subspace of $R$.

\begin{thm}\label{tm:zadnjidio-2}
Let \Xdalfa\ be a computable metric space and let $S$ be a semicomputable graph in this space. Let $\{A_0,\dots,A_n\}$ be a family of arcs and topological rays which defines $S$. Let $\Gamma $ be the set of all endpoints of $S$ which are uncomputable. Let $\eps >0$. Then there exist arcs and topological rays $A_0',\dots,A_n'$ such that for each $i\in \{0,\dots,n\}$ the following hold:
\begin{enumerate}[label=(\roman*)]
    \item if $A_i $ is an arc whose endpoints do not belong to $\Gamma $, then $A_i'=A_i $;
    \item if $A_i $ is an arc and $x$ and $y$ are endpoints of $A_i $ such that $x\in \Gamma $ and $y\notin \Gamma $, then $A_i'=\luk zy{A_i}$, where $z\in A_i $ and $\luk xz{A_i}\subseteq B(x,\eps) $;
    \item if $A_i $ is an arc and $x$ and $y$ are endpoints of $A_i $ such that $x,y\in \Gamma $,  then $A_i'=\luk{z_1}{z_2}{A_i}$, where $z_1,z_2 \in A_i$ and $\luk{x}{z_1}{A_i}\subseteq B(x,\eps)$,  $\luk{z_2}{y}{A_i}\subseteq B(y,\eps)$.
    \item if $A_i $ is a topological ray whose endpoint does not belong to $\Gamma $, then $A_i'=A_i $;
    \item if $A_i $ is a topological ray whose endpoint $x$ belongs to $\Gamma $, then  $A_i'=(A_i\setminus  \luk xz{A_i})\cup \{z\}$, where $z\in A_i $ and $\luk xz{A_i}\subseteq B(x,\eps) $;
    \item the set $T=A_0'\cup \dots \cup A_n'$ is a computable graph with computable endpoints.
\end{enumerate}
\end{thm}
\begin{proof}
We proceed as in the proof of Theorem~\ref{tm:zadnjidio}.

Let $A_{0}' $ be defined in the same way as in the proof of Theorem~\ref{tm:zadnjidio} in the case when $A_{0} $ is an arc. 

If $A_{0} $ is a ray whose endpoint does not belong to $\Gamma $, let $A_{0}'=A_{0} $. 

If $A_{0} $ is a ray whose endpoint $x$ belongs to $\Gamma $, then $A_{0} '$ will be defined as $(A_i\setminus  \luk xz{A_i})\cup \{z\}$, where $z\in A_i $, $\luk xz{A_i}\subseteq B(x,\eps) $ and $A_{0} '\cup A_{1} \cup \dots \cup A_{n}$ is a semicomputable set. We obtain such a ray $A_{0} '$ in the same way as we obtained $A_{0} '$ in the proof of Theorem~\ref{tm:zadnjidio} in the case when $A_{0} $ is an arc whose one endpoint belongs to $\Gamma $.

In finitely many steps we get $A_0',\dots,A_n'$ which satisfy properties (i)--(v) such that $\{A_0',\dots,A_n'\}$ defines a semicomputable graph $T$ which has no uncomputable endpoints. Hence, every endpoint of $T$ is computable, so $T$ is computable by Theorem~\ref{tm:grafovi}.
\end{proof}

Although graphs in metric spaces are unbounded in general and therefore the Hausdorff distance does not make sense in this context, we do have the following immediate consequence of Theorem~\ref{tm:zadnjidio-2}.
\begin{cor}\label{cor:zadnjidio}
Let \Xdalfa\ be a computable metric space and let $S$ be a semicomputable graph in this space. Then for each $\eps >0$ there exists a computable graph $T$ in \Xdalfa\ such that $T\subseteq S$, all endpoints of $T$ are computable and $S\approx_\eps\! T$.
\end{cor}

\subsection{1-manifolds}

A \defn{$1$-manifold with boundary} is a second-countable Hausdorff space $M$ such that each point of $M$ has a neighbourhood homeomorphic to $[0,\infty\rangle $. The \defn{boundary} $\partial M$ of a such an $M$ is defined as the set of all $x\in M$ such that $x$ has a neighbourhood in $M$ homeomorphic to $[0,\infty\rangle $ by a homeomorphism which maps $x$ to $0$.

If $M$ is a $1$-manifold with boundary, then each component of $M$ is homeomorphic to one of these spaces~\cite{shastri}:
\[\R,\quad [0,\infty\rangle,\quad \Sjedan\text{,\quad or\quad}[0,1].\]
We have the following: $\R$ is the union of two rays in $\R$ which intersect in a common endpoint, $[0,\infty\rangle $ is a ray in $\R$,  $\Sjedan$ is homeomorphic to the union of any three line segments in $\R^2$ each two of which intersect in a common endpoint, and $[0,1]$ is a line segment in $\R$. In general, each component is a closed set. So if we additionally assume that $M$ has finitely many components, it follows that $M$ is a graph. Hence, each $1$-manifold with boundary $M$ with finitely many components is a graph and  $\partial M$ is the set of all endpoints of the graph $M$. We have the following consequence of Corollary~\ref{cor:zadnjidio}.
\begin{cor} 
Let \Xdalfa\ be a computable metric space and let $M$ be a semicomputable 1-manifold in this space such that $M$ has finitely many components. Then for each $\eps >0$ there exists a computable 1-manifold $N$ in \Xdalfa\ such that $N\subseteq M$, each point of $\partial N$ is computable and $M\approx_\eps\! N$.
\end{cor}

\paragraph*{Acknowledgments}
This paper was supported by the European Union---NextGenerationEU through the National Recovery and Resilience Plan 2021--2026 Institutional grant of University of Zagreb, Faculty of Science (IK IA 1.1.3. Impact4Math).

\bibliographystyle{alphaurl} 
\bibliography{main.bib,References.bib}
\end{document}